\documentclass[11pt]{article}
\usepackage[utf8]{inputenc}
\usepackage[T1]{fontenc}
\usepackage{natbib}
\usepackage{xurl}
\usepackage{doi}

\usepackage[final]{microtype}
\usepackage{amsmath,amssymb,amsfonts,amsthm,mathtools}
\usepackage{newtxtext,newtxmath}
\usepackage{bm}
\usepackage{booktabs}
\usepackage{array}
\usepackage[dvipsnames]{xcolor}
\usepackage{enumitem}
\usepackage{geometry}
\usepackage[nameinlink,capitalise,noabbrev]{cleveref}
\usepackage{comment}
\usepackage[nameinlink]{cleveref}
\setlist[itemize]{leftmargin=*,topsep=3pt,itemsep=2pt,parsep=0pt}
\setlist[enumerate]{leftmargin=*,topsep=3pt,itemsep=2pt,parsep=0pt}
\definecolor{LinkBlue}{RGB}{18,63,120}
\newcolumntype{L}[1]{>{\raggedright\arraybackslash}p{#1}}
\hypersetup{
  colorlinks=true,
  linkcolor=LinkBlue,
  citecolor=LinkBlue,
  urlcolor=LinkBlue,
  pdfborder={0 0 0}
}
\theoremstyle{plain}
\newtheorem{theorem}{Theorem}[section]
\newtheorem{proposition}[theorem]{Proposition}
\newtheorem{lemma}[theorem]{Lemma}
\newtheorem{corollary}[theorem]{Corollary}
\theoremstyle{definition}
\newtheorem{definition}[theorem]{Definition}

\theoremstyle{remark}
\newtheorem{remark}[theorem]{Remark}
\crefname{theorem}{theorem}{theorems}
\Crefname{theorem}{Theorem}{Theorems}
\crefname{proposition}{proposition}{propositions}
\Crefname{proposition}{Proposition}{Propositions}
\crefname{lemma}{lemma}{lemmas}
\Crefname{lemma}{Lemma}{Lemmas}
\crefname{corollary}{corollary}{corollaries}
\Crefname{corollary}{Corollary}{Corollaries}
\crefname{definition}{definition}{definitions}
\Crefname{definition}{Definition}{Definitions}
\crefname{example}{example}{examples}
\Crefname{example}{Example}{Examples}
\crefname{remark}{remark}{remarks}
\Crefname{remark}{Remark}{Remarks}

\newcommand{\bw}{\bm w}
\newcommand{\bq}{\bm q}
\newcommand{\bx}{\bm x}
\newcommand{\bmu}{\bm\mu}
\newcommand{\bzero}{\bm 0}
\newcommand{\E}{\mathbb E}
\newcommand{\Prob}{\mathbb P}
\newcommand{\OPT}{\mathrm{OPT}}
\newcommand{\DEC}{\mathrm{DEC}}
\newcommand{\REFone}{\mathrm{REF}_1}

\newcommand{\USW}{\mathrm{SW}}
\newcommand{\eps}{\varepsilon}
\newcommand{\ind}{\mathbf 1}

\title{Approximating Optimal Welfare in Complementary Allocation\\ under Decentralized Information}
\author{Meryem Essaidi}
\date{}

\hypersetup{
  colorlinks=true,
  pdfborder={0 0 0},
  pdftitle={Approximating Optimal Welfare in Complementary Allocation under Decentralized Information},
  pdfauthor={Meryem Essaidi},
  pdfsubject={Decentralized allocation of complementary resources under local information},
  pdfkeywords={algorithmic game theory, communication, complementary allocation, decentralized information, threshold protocols}
}

\begin{document}

\maketitle

\begin{abstract}
Complementary resources are often allocated by agencies that can only see parts of an individual’s need profile; and thus may know that an individual lacks its own resource without knowing whether supplying it will complete a useful bundle. We study allocation of $m$ divisible complementary resources when each agency observes only its own coordinate of a baseline-access profile (i.e., only one of the individual's $m$ coordinates.) 

To isolate information from incentives, we measure the welfare cost of this information split against a deliberately generous decentralized benchmark \textbf{DEC}: the best allocation rule that acts on local information alone, with fully cooperative agencies, and known population distribution and ex-ante resource capacities. Even against that benchmark, local information is arbitrarily costly: the best local-product allocation can lose a \textbf{linear-factor $\Theta(m)$} in the number of agencies, even under equal capacities; or an \textbf{inverse-optimum factor} $\Theta(1/\OPT)$, with just three agencies. 

Threshold referrals recover much of this loss: each agency reports if its coordinate lies below a public target, and a clearing rule allocates using only the joint ($m$-bit) reports. {\itshape For the algorithmic reader, in communication-complexity terms, the input is partitioned by coordinate and the message language is one round of threshold bits.} We show that aggregating these profile-types by their reports implements a rectangle under the optimal half-utility survival-curve giving a log-approximation of the fully-informed centralized optimum \textbf{OPT}. Our main tool is a profile-level charging certificate: if optimal utility is at least twice a service threshold, the optimum pays for that threshold on every deficient coordinate. This half-utility charging argument yields welfare at least 
$$\frac{\OPT}{4(1+\ln(2/\OPT))}$$ with one uplink bit per agency. 

Viewed as a mechanism-design question, we ask how much of the loss a coarse message can recover: The extremal survival-curve (\emph{equal-revenue family}) makes this tight, supplying the matching lower-bounds and an exact expression for the welfare attainable with any fixed number of ordered levels r, namely $\Theta(\log r)$ bits. Reporting a bucket in a geometric threshold-ladder thus improves this to a constant fraction of optimal welfare with doubly-logarithmic message-length: a staircase recovers $\OPT/8$ with $\Theta(\log\log(1/\OPT))$ bits per agency, and this is necessary within this class.\\ 

For a planner, the message is short: 
\par
{\centering\itshape
One threshold guarantees a logarithmic approximation;\\
a staircase with log-of-log bits, a constant-factor approximation.\\
\par}

$ $\\
These are tight within the model, and the extremal family admits an exact (optimal) welfare formula. Thresholds are chosen offline, and communication is limited to one-round. These results isolate \emph{how a simple reporting language converts severe local-info losses into constant-factor recovery}; thus quantifying the value of coarse communication; without agencies revealing their data.
\end{abstract}

\section{Introduction}
\label{sec:introduction}

Allocation across agencies is often decentralized precisely where resources are complementary. One agency controls housing-support, another medical care, a third job-placement, and a fourth food-assistance; yet a person who needs several of these resources benefits only when those needs are met together. Food assistance alone may generate little welfare for a severely ill person whose unmet medical need remains the binding constraint. 

Because the effectiveness of each intervention hinges on matching support in the other deficient coordinates, welfare is governed by the individual's smallest-endowed coordinate. For each individual, the natural allocation principle is therefore to water-fill the baseline-access levels from lowest to highest.

Agencies can share a common objective and know the population distribution, yet each sees only its own local coordinate. The induced loss is informational rather than strategic, and we isolate that loss in a minimal model. 

A type \(\bw=(w_1,\ldots,w_m)\in[0,1]^m\) records baseline access to \(m\) divisible resources. Agency \(j\) is constrained by an ex-ante supply \(q_j\) and controls only top-ups $x_j$ of coordinate \(j\); and can \emph{only} see \(w_j\). An individual with baseline-access $\bw$ receiving tops-ups $\bx$ capped \textit{Leontief-utility} welfare: 
\[u(\bw,\bx)=\min\{1,w_1+x_1,\ldots,w_m+x_m\}.\]
The benchmark \(\OPT\) observes each individual's full baseline-access profile \(\bw\). The local-information benchmark \(\DEC\) optimizes over every distribution-aware product rule \(x_j=x_j(w_j)\). Consequently, any lower-bound against \(\DEC\) persists even in an otherwise frictionless, non-strategic setting; with cooperative agencies that are both prior-aware and \emph{computationally unbounded.} Abstracting from these incentive and computational frictions isolates the \emph{irreducible welfare loss} due solely to decentralized information.

The only restriction is that each agency observes only its own realized coordinate.

The core question is then quantitative: \textbf{\emph{how much structured coarse communication is needed to recover full-information welfare \(\OPT\)?}} We study an operationally simple language of threshold-referrals. For a public target \(\mu\), agency \(j\) sends the bit \(\ind\{w_j<\mu\}\). The clearing rule sees the joint $m$-bit profile (but not raw coordinates), and chooses if individual is to be served at level \(\mu\). If served, each agency uses then supplies the top-up \((\mu-w_j)_+\). With an ordered ladder of targets, an agency sends only the bucket containing its coordinate and the clearing rule selects one ladder level.

\paragraph{Note on capacity-accounting and implementation:}
For analysis, we conservatively charge \(\mu\) rather than the smaller realized top-up, to make capacity-accounting \emph{independent of raw coordinates.} This coarse reporting language yields several operational gains. An agency need not transmit nor retain the magnitude of an individual's deficit, nor condition later decisions on it. It need only know if local baseline falls below \(\mu\). Since the threshold bit  (flagged as $1$ if below $\mu$) by itself suffices for capacity-use accounting, the resulting standardized bit reports reduce both centralized collection and agency-level retention of individual-level data; all the while allowing a clearing decision to be \emph{executed immediately and locally.} As to why this makes feasibility transparent: conditional on a reported profile, each admitted flagged coordinate is charged \(\mu\), independently of an individual's exact baseline. Capacity accounting is therefore \emph{baseline-agnostic} within each profile and reduces to a simple profile-level \textbf{linear-program.}

\subsection{Results}

The results separate the loss from local observation from the welfare recoverable through threshold messages.
\begin{center}
\footnotesize
\begin{tabular}{@{}
L{0.17\linewidth}
L{0.17\linewidth}
L{0.26\linewidth}
L{0.27\linewidth}@{}}
\toprule
%%%%%%%%%%%%%%%%%%%%%%%%
\textbf{Mechanism / Class} 
%%%%%%%%%%%%%%%%%%%%%%%%
& \textbf{Per-agency bits} 
& \textbf{Worst-case loss ratio} 
& \textbf{Tightness / exact families}
\\
\midrule
%%%%%%%%%%%%%%%%%%%%%%%%
Best local-product
%%%%%%%%%%%%%%%%%%%%%%%%
& $\displaystyle 0$
& \(\displaystyle\frac{\OPT}{\DEC}\geq\Omega(\max\{m,\mathbin{1\slash\OPT}\})\)
& \(\left.\begin{array}{l}
{\displaystyle\Theta(m)}; \textit{equal capacities.}\\
{\displaystyle\Theta(\mathbin{1\slash\OPT})}; \textit{three agencies.}\\[0.5em]
\end{array}\right.\)\\
%%%%%%%%%%%%%%%%%%%%%%%%
One threshold \(\displaystyle\mu^*\)
%%%%%%%%%%%%%%%%%%%%%%%%
& $\displaystyle 1$
& \(\displaystyle\frac{\OPT}{\REFone}\leq 4(1+\ln\mathbin{(2/\OPT)})\)
& \(\left.\begin{array}{l}
\textit{Equal-revenue;}\\
{\displaystyle1+\ln(\mathbin{1\slash\eps})}, \textbf{exact.}\\[0.5em]
\end{array}\right.\)\\
%%%%%%%%%%%%%%%%%%%%%%%%
\small Dyadic ladder \(\displaystyle\bmu\)
%%%%%%%%%%%%%%%%%%%%%%%%
& \(\displaystyle\Theta(\log\log(\mathbin{1\slash\OPT}))\)
& \(\displaystyle\frac{\OPT}{R(\bmu)}\leq 8\)
& \(\displaystyle\left.\begin{array}{l}
\text{Matching bit-order;}\\
\text{within this class.}\\[0.5em]
\end{array}\right.\)\\
\bottomrule
%%%%%%%%%%%%%%%%%%%%%%%%
\end{tabular}
\end{center}

\paragraph{Local information can be arbitrarily weak.}
The first family has \(m\) target types, each missing a different coordinate, and one all-zero decoy. Agency \(j\) cannot distinguish its target from the decoy because both have \(w_j=0\). Centralization obtains welfare \(mq\), whereas the exact local optimum is \(q/[1-(m-1)q]\), producing a \(\Theta(m)\) gap. A second construction uses only three agencies: two cannot distinguish target from decoy, while the third has scarce supply. Its exact local optimum yields a \(\Theta(1/\OPT)\) gap. These are lower bounds against the best local policy, not examples of myopic behavior.

\paragraph{One common threshold \(\mu\) gives logarithmic recovery.}
Fix a full-information optimum and let \(U=u^{\OPT}/2\). If \(U(\bw)\ge\mu\) but \(w_j<\mu\), then the optimum must spend more than \(\mu\) in coordinate \(j\). As such, the mass of high-\(U\) types in any referral profile gives a feasible solution to a conservative profile-admission LP.  At any \(\mu\), this solution certifies welfare \(\mu\Pr[U\ge\mu]\). Truncating \(U\) below \(\OPT/4\) loses at most \(\OPT/4\), while the best rectangle under the remaining survival curve loses \textbf{only a logarithmic factor}; proving the guarantee.

The logarithm can't be removed by choosing a better single threshold. In a two-agency family, types are \((s,0)\), the first agency has zero supply, and \(s\) has an equal-revenue tail. Full information spends \(s\) only in the second coordinate and obtains \(\E s\). A threshold above \(s\) would require forbidden spending in coordinate one; a threshold below \(s\) contributes at most the equal-revenue rectangle \(\mu\Pr[s\ge\mu]\). \emph{Every threshold therefore obtains at most \(\eps\), and this value is attained, while \(\OPT=\eps(1+\ln(1/\eps))\).}

\paragraph{Ordered geometric ladders give a tight doubly-logarithmic bit frontier.}
For levels \(\mu_1<\cdots<\mu_r\), the same charging argument implements the lower staircase \(g(U)=\max\{\mu_\ell:\mu_\ell\le U\}\). Choosing a dyadic ladder from \(\OPT/4\) to \(1/2\) ensures \(g(U)\ge U/2\) on the retained range. Truncation retains welfare at least \(\OPT/4\), and geometric rounding-down loses at most a factor of two; thus certifying \textbf{welfare at least \(\OPT/8\).}

There are \(r\!\!=\!\!\Theta(\log(1/\OPT))\) levels but only \(r\!+\!1\) local messages, so the uplink message length is \(\Theta(\log\log(1/\OPT))\) bits per agency. On the equal-revenue family, each ordered service level contributes at most \(\eps\), which forces \(r\!=\!\Omega(\log(1/\OPT))\) for constant-factor recovery. In fact, \emph{the optimal ladder on this family is geometric between \(\eps\) and one, and its welfare has an exact closed form} (\Cref{cor:exact-ladder}).

\subsection{Positioning and scope}

Our price ratio is in the tradition of worst-case efficiency measures \citep{koutsoupias1999worst,papadimitriou2001algorithms,bertsimas2011price}, but the constraint here is decentralized observation. This connects to team decision problems with dispersed information \citep{marschak1972economic} and to communication requirements in allocation \citep{nisan2006communication,blumrosen2007severely,dobzinski2019economic}. The mechanism is also related in spirit to simple posted instruments \citep{myerson1981optimal,hartline2009simple}: a threshold turns a complex distribution into a rectangle under a survival curve.

Recent work uses threshold queries to elicit agents' preferences in one-sided matching \citep{ma2021improving,latifian2024distortion} and studies allocation from partial rankings \citep{halpern2021fair}. Here, by contrast, no single agent holds the full type: different agencies observe different coordinates of one individual's baseline, and complementary welfare and resource capacities couple their decisions. Multi-resource allocation is central in dominant-resource, market-equilibrium, and Leontief models \citep{ghodsi2011dominant,parkes2015beyond,eisenberg1959consensus,bei2022fair}; our distinguishing restriction is coordinate-split observation. One-bit congestion broadcasts have also been used for iterative distributed resource allocation \citep{alam2018communication}; our messages instead travel once, upward from locally informed agencies to a clearing rule. Finally, unlike work on jointly private allocation \citep{hsu2016private}, the threshold messages here carry no formal privacy guarantee.

\paragraph{Information structure and scope.}
Public primitives \((F,\bq)\) are fixed offline. Thus \(\DEC\) gives each coordinate rule full knowledge of the environment; withholding only an individual's other realized coordinates. The lower bounds in \Cref{thm:local-loss} thus isolate coordinate-local observation as the informational bottleneck (rather than computation or incentives); and these bounds continue to hold when local rules have less prior information. %

Agencies observe their coordinates and execute prescribed threshold maps. The trusted clearing rule receives an \(m\)-bit referral profile for one threshold, or one bucket index from each agency for a threshold ladder; it never sees raw coordinates. 

Its admission kernel is chosen offline and is distribution-dependent: it computes profile-masses and admission-probabilities using \((F,\bq)\). (which may, for instance, be estimated from historical administrative data). \textbf{The paper takes them as given and isolates the welfare consequences of decentralized observation.} Estimation of \(F\) from finite data, uncertainty about capacities, strategic communication, interaction, and cryptographic privacy guarantees lie outside the model.

The resulting frontier applies to one-round common-threshold and ordered-threshold protocols under ex-ante supplies. A companion scalar-ceiling result extends this to unrestricted finite-transcript protocols: after padding the protocol tree, a \(b\)-bit protocol has at most \(2^b\) effective transcript cells, while \emph{indistinguishable decoy-padding with auxiliary coordinates conveys no additional information} about the active scalar parameter.

\subsection{Organization}

\Cref{sec:model} specifies the stochastic protocol and its communication accounting. \Cref{sec:local} gives the local-information lower bounds. \Cref{sec:one-threshold} proves the one-bit guarantee and its exact equal-revenue obstruction. \Cref{sec:ladders} proves the ordered-ladder frontier. \Cref{sec:implementation} records computation and finite-population scope. Full proofs appear in \Cref{app:proofs}.
\section{Model and protocol semantics}\label{sec:model}

Consider \(m\ge2\) agencies and a unit mass of individuals. A type \(\bw=(w_1,\ldots,w_m)\) is drawn from a public Borel distribution \(F\) on \([0,1]^m\). Agency \(j\) observes only \(w_j\) and controls resource \(j\), whose per-capita supply is \(q_j\ge0\).

An allocation is a measurable top-up map \(X:[0,1]^m\to\mathbb R_+^m\). It is feasible if
\begin{equation}
   0\le X_j(\bw)\le 1-w_j
   \quad\text{and}\quad
   \E_F[X_j(\bw)]\le q_j
   \qquad(j\in[m]).
   \label{eq:feasibility}
\end{equation}
The pointwise upper bound is without loss as service beyond one can't increase capped utility. Welfare is
\[u^X(\bw)=\min\{1,\min_{j\in[m]}(w_j+X_j(\bw))\},
\qquad\USW(X)=\E_F[u^X(\bw)].\]

\subsection{Information benchmarks}

The full-information and local-product values are
\begin{equation}
  \OPT=\sup_{X\ \mathrm{feasible}}\USW(X),
  \qquad
  \DEC=
  \sup_{\substack{X\ \mathrm{feasible}:\\
                    X_j(\bw)=x_j(w_j)\ \forall j}}
       \USW(X).
  \label{eq:benchmarks}
\end{equation}
The one-coordinate functions \(x_j\) may be chosen with complete knowledge of \((F,\bq)\); only their realized input is local. We use \(\OPT/\DEC\) as the price of local information, with the usual value \(+\infty\) when \(\OPT>0=\DEC\).

Deterministic local-product rules suffice even if one initially permits type-independent private or shared randomness.

\begin{lemma}[Interim reduction]
\label{lem:interim}
For every feasible randomized local-product rule there is a deterministic local-product rule with the same expected coordinate loads and weakly larger expected welfare.
\end{lemma}

\subsection{One-round threshold protocols}

We state the randomization and feasibility semantics explicitly. Public primitives \((F,\bq)\) and all thresholds are fixed offline. For each individual, agency \(j\) sends a message determined only by \(w_j\). A clearing rule observes the joint message and a random seed \(R\) independent of \(\bw\), and draws one service level. The selected level is broadcast to the agencies. If level \(\mu\) is selected, agency \(j\) uses its own \(w_j\) to supply \((\mu-w_j)_+\). Feasibility is ex ante over both \(\bw\) and \(R\).

For one threshold \(\mu\in(0,1]\), agency \(j\) sends
\[a_j^\mu(w_j)=\ind\{w_j<\mu\},
\qquad
S^\mu(\bw)=\{j:w_j<\mu\},
\qquad
p_S^\mu=\Prob[S^\mu(\bw)=S].\]
The clearing rule admits a type with profile \(S\) with a probability depending only on \(S\). Conservative accounting charges \(\mu\), rather than the smaller actual top-up \(\mu-w_j\), to each flagged coordinate. The resulting profile LP is
\begin{align}
 R(\mu)=\max_{(z_S)}\quad
   &\mu\sum_{S\subseteq[m]}z_S
   \label{lp:one-objective}\\
 \text{s.t.}\quad
   &0\le z_S\le p_S^\mu
       &&(S\subseteq[m]),\label{lp:one-mass}\\
   &\mu\sum_{S\ni j}z_S\le q_j
       &&(j\in[m]).\label{lp:one-capacity}
\end{align}
Here \(z_S\) is admitted population mass, not a probability. Conditional admission probability \(z_S/p_S^\mu\) implements any feasible LP solution, and the actual load is at most the conservative charge. Define
\[\REFone=\sup_{0<\mu\le1}R(\mu).\]
The LP objective certifies welfare from the selected service levels; it does not count baseline utility on unserved types or baseline utility above a selected level. Actual welfare is therefore at least the LP objective, not necessarily equal to it. In particular, \(\REFone\) denotes the best such certificate, not the optimum under actual-cost accounting.

For an ordered ladder \(0<\mu_1<\cdots<\mu_r\le1\), set \(\mu_0=0\) and let agency \(j\) send the bucket
\[a_j^{\bmu}(w_j)=\max\{\ell\in\{0,\ldots,r\}:\mu_\ell\le w_j\},\]
where the maximum is \(0\) if \(w_j<\mu_1\). Thus the joint message reveals, for every level \(\ell\), exactly the referral set
\[S_\ell(a)=\{j:a_j<\ell\}=\{j:w_j<\mu_\ell\}.\]
Let \(p_a=\Prob[a^{\bmu}(\bw)=a]\). Variables \(z_{a\ell}\) assign mass from message profile \(a\) to service level \(\mu_\ell\):
\begin{align}
 R(\bmu)=\max_{(z_{a\ell})}\quad
   &\sum_a\sum_{\ell=1}^r\mu_\ell z_{a\ell}
   \label{lp:ladder-objective}\\
 \text{s.t.}\quad
   &z_{a\ell}\ge0,\qquad
     \sum_{\ell=1}^r z_{a\ell}\le p_a
       &&(a\in\{0,\ldots,r\}^m),\label{lp:ladder-mass}\\
   &\sum_a\sum_{\ell:j\in S_\ell(a)}
        \mu_\ell z_{a\ell}\le q_j
       &&(j\in[m]).\label{lp:ladder-capacity}
\end{align}
The corresponding kernel chooses level \(\ell\) with conditional probability \(z_{a\ell}/p_a\). This proves both implementability and ex-ante feasibility of the LP.

We count fixed-length uplink messages from agencies to the clearing rule, using the full prescribed bucket alphabet. One threshold uses one bit per agency. An \(r\)-level ladder has \(r+1\) local buckets and uses \(\lceil\log_2(r+1)\rceil\) bits per agency. These are worst-case message budgets, not instance-wise lower bounds for every agency: an agency whose coordinate is constant can communicate less if unused buckets are removed. Threshold descriptions, the offline kernel, and individual identifiers are public infrastructure and are not charged. If round-trip communication is counted, the clearing rule additionally broadcasts one symbol from the \(r+1\) service-level alphabet.

\begin{remark}[What the messages do and do not reveal]
The clearing rule observes the joint threshold or bucket profile, not the raw coordinates. This is an information restriction, not a formal privacy guarantee: no differential, cryptographic, or strategic privacy claim is made.
\end{remark}
\section{The price of local information}
\label{sec:local}

Two exact families show that local-product information can be much weaker than full information even when all agencies know \((F,\bq)\) and optimize the same objective. 

\begin{definition}[Balanced target--decoy family]\label{def:balanced} 
Fix \(m\ge2\) and \(q\in(0,1/(2m)]\).  Each agency has supply \(q\). For each \(i\in[m]\), type \(\bw^{(i)}\) has mass \(q\), coordinate \(w_i^{(i)}=0\), and all other coordinates equal to one.  The remaining mass \(1-mq\) is the all-zero decoy. 
\end{definition} 

Agency \(i\) sees the same local value on its unique target and on the decoy. Full information can distinguish them and spend resource \(i\) only on the target.

\begin{definition}[Scarce-complement family]\label{def:scarce}
Fix \(q\in(0,1/2)\) and \(q_{\min}\in(0,q]\).  There are three agencies with supplies \((q,q,q_{\min})\).  The target \((0,0,1)\) has mass \(q\), and the all-zero decoy has mass \(1-q\).
\end{definition}

\begin{theorem}[Exact local-information losses]\label{thm:local-loss} 
The following identities hold, also for randomized local-product rules with type-independent randomness. 
    \begin{enumerate} 
        \item On the balanced family, 
        \[\OPT=mq,\qquad\DEC=\frac{q}{1-(m-1)q},\qquad\frac{\OPT}{\DEC}=m\bigl(1-(m-1)q\bigr)\in\left[\frac{m+1}{2},m\right).\] 
        \item On the scarce-complement family, 
        \[\OPT=q,\qquad \DEC=q^2+\min\{q_{\min},q(1-q)\}.\] 
        In particular, if \(q_{\min}\le q^2\), then \(\OPT/\DEC=\Theta(1/q)=\Theta(1/\OPT)\). 
    \end{enumerate} 
\end{theorem}

The two examples separate dimensional dilution from scarcity.  The balanced loss grows with the number of complementary resources even under equal supplies.  The scarce-complement loss grows like the reciprocal welfare scale with only three agencies.   In both cases, the obstruction is the same: targets and decoys occupy a common local information cell.

\begin{remark}[Combined worst-case loss]
\label{rem:combined-local-loss}
Padding the scarce-complement family with agencies whose coordinates are identically one leaves both \(\OPT\) and \(\DEC\) unchanged. Combining the resulting \(\Omega(1/\OPT)\) construction with the balanced target--decoy family's \(\Omega(m)\) construction shows that, for every \(m\ge3\) and \(0<\OPT\le1/2\), there exists an instance satisfying
\[
    \frac{\OPT}{\DEC}
    =
    \Omega\!\left(\max\left\{m,\frac{1}{\OPT}\right\}\right).
\]
\end{remark}
\section{One-threshold referrals}
\label{sec:one-threshold}

The referral profile separates many of the local information cells that cause \Cref{thm:local-loss}. The central proof device is a coordinate-wise charging certificate.

Fix a full-information optimal allocation \(X^\star\), write \(u^\star=u^{X^\star}\), and define the half-utility random variable
\[U(\bw)=\frac{u^\star(\bw)}{2}\in[0,1/2].\]
Attainment of the optimum is justified at the start of \Cref{app:proofs}.

\begin{lemma}[Half-utility charging]
\label{lem:charging}
For every \(\mu\in(0,1]\), the one-threshold LP satisfies
\[R(\mu)\ge\mu\Prob[U\ge\mu].\]
\end{lemma}

Indeed, assign to each referral profile exactly its mass of types satisfying \(U\ge\mu\). Whenever coordinate \(j\) is flagged, the optimum must already spend more than \(\mu\) on that coordinate:
\[U\ge\mu,\quad w_j<\mu\quad\Longrightarrow\quad X_j^\star\ge u^\star-w_j\ge2\mu-w_j>\mu.\]
The optimum's capacity therefore pays for every conservative LP charge.

\begin{theorem}[One-bit logarithmic recovery]
\label{thm:one-bit}
For every instance with \(0<\OPT\le1\), some common threshold \(\mu\in[\OPT/4,1/2]\) implements a one-bit-per-agency protocol satisfying
\[\REFone\ge\frac{\OPT}{4\bigl(1+\ln(2/\OPT)\bigr)}.\]
\end{theorem}

The logarithm is the cost of replacing the area under a survival curve by its best anchored rectangle. This is not an artifact of the proof.

\begin{definition}[Equal-revenue referral family]
\label{def:equal-revenue}
Fix \(\eps\in(0,1)\). There are two agencies, types have the form \(\bw=(s,0)\), and
\[\Prob[s\ge t]=
 \begin{cases}
   1, & 0\le t\le\eps,\\
   \eps/t, & \eps<t\le1,\\
   0, & t>1.
 \end{cases}
 \label{eq:equal-revenue-tail}\]
Equivalently, \(s\) has density \(\eps/t^2\) on \([\eps,1)\) and an atom of mass \(\eps\) at one. Supplies are
\[q_1=0,
\qquad
q_2=\E[s]=\eps\bigl(1+\ln(1/\eps)\bigr).\]
\end{definition}

\begin{theorem}[Exact one-threshold obstruction]
\label{thm:one-threshold-lower}
On the equal-revenue family,
\[\OPT=\eps\bigl(1+\ln(1/\eps)\bigr),
  \qquad
  \sup_{\textnormal{single-threshold protocols}}\USW=\eps.\]
Thus the best single-threshold protocol has exact approximation ratio \(1+\ln(1/\eps)\). The upper bound holds under actual top-up accounting, and therefore also under conservative accounting.
\end{theorem}

The zero first-coordinate supply makes the obstruction transparent. A type with \(s<\mu\) cannot be raised to \(\mu\), while types with \(s\ge\mu\) have total rectangle value \(\mu\Prob[s\ge\mu]\le\eps\). Full information instead matches the second-coordinate top-up to each realized \(s\).
\section{Ordered threshold ladders}
\label{sec:ladders}

A ladder replaces one rectangle by a lower staircase. For \(0<\mu_1<\cdots<\mu_r\le 1\), define
\[g_{\bmu}(u)=\max\bigl(\{\mu_\ell:\mu_\ell\le u\}\cup\{0\}\bigr).\]

\begin{lemma}[Staircase certificate]
\label{lem:staircase}
For every ordered ladder,
\[R(\bmu)\ge\E[g_{\bmu}(U)].\]
\end{lemma}

The certificate assigns each type to the largest ladder level below its half-utility. Each type is assigned once. If agency \(j\) is flagged at that level, the same pointwise inequality as in \Cref{lem:charging} charges the conservative cost to \(X_j^\star\). Since \(U\le 1/2\), levels above \(1/2\) are never used by this certificate. Such levels are permitted in an ordered ladder and cannot reduce welfare, since the clearing rule may assign them zero probability.

\begin{theorem}[Constant recovery with a geometric ladder]
\label{thm:geometric-ladder}
For an instance with \(0<\OPT\le1\), let \(\delta=\OPT/4\),
\[
  L=\left\lceil\log_2\!\left(\frac{2}{\OPT}\right)\right\rceil,
  \qquad
  r=L+1,
  \qquad
  \mu_\ell=\min\{2^{\ell-1}\delta,1/2\}
  \quad(\ell=1,\ldots,r).
\]
These levels are strictly increasing and the resulting protocol satisfies
\[R(\bmu)\ge\frac{\OPT}{8}.\]
It uses
\[b=\left\lceil\log_2(r+1)\right\rceil=O\bigl(\log\log(1/\OPT)\bigr)\]
uplink bits per agency as \(\OPT\downarrow0\).
\end{theorem}

The construction retains at least \(\OPT/4\) of \(\E U=\OPT/2\) above \(\delta\). Doubling ensures \(g_{\bmu}(u)\ge u/2\) throughout that retained range. (Its top level is $1/2$, so no level exceeds the range of $U$.)

The bit order is necessary within ordered-threshold protocols.

\begin{theorem}[Ordered-threshold lower bound]
\label{thm:ladder-lower}
On the equal-revenue family of \Cref{def:equal-revenue}, every feasible ordered protocol with \(r\) positive service levels has welfare at most
\[r\eps.\]
Consequently, any ordered-threshold protocol that obtains a fixed positive fraction of \(\OPT\) on every instance requires
\[r=\Omega\bigl(\log(1/\OPT)\bigr)
   \quad\text{levels and}\quad
   b=\Omega\bigl(\log\log(1/\OPT)\bigr)
   \quad\text{uplink bits per agency}\]
as \(\OPT\downarrow0\).
\end{theorem}

Combining \Cref{thm:geometric-ladder,thm:ladder-lower} gives a tight \(\Theta(\log\log(1/\OPT))\) per-agency bit frontier for constant-factor recovery within the ordered-threshold class.

The same example admits a sharper, finite-level characterization.

\begin{corollary}[Exact ladder value on the equal-revenue family]
\label{cor:exact-ladder}
Let \(V_r(\eps)\) be the largest welfare attainable on \Cref{def:equal-revenue} by an ordered protocol with at most \(r\) positive service levels in \((0,1]\). Under either actual or conservative top-up accounting, \(V_1(\eps)=\eps\), and, for \(r\ge2\),
\[V_r(\eps)=\eps\left[1+(r-1)\left(1-\eps^{1/(r-1)}\right)\right].\]
An optimal ladder is \(\mu_\ell=\eps^{(r-\ell)/(r-1)}\), \(\ell=1,\ldots,r\).
\end{corollary}

\begin{proof}
For a fixed ladder, feasibility forces the selected level to be at most \(s\). Conversely, assigning the largest level at most \(s\) is a bucket-measurable rule. It uses no resource one and spends \(g_{\bmu}(s)\le s\) in resource two, so is feasible under either accounting convention. The optimal value for that ladder is therefore \(\E g_{\bmu}(s)\).

Only the largest level below \(\eps\) can ever be selected; raising it to \(\eps\) weakly improves the value, after deleting any duplicate. Thus levels below \(\eps\) can be removed without loss. Adding levels never decreases welfare, so it suffices to optimize over exactly \(r\) distinct levels in \([\eps,1]\). The one-level case is \Cref{thm:one-threshold-lower}. For \(r\ge2\), the tail formula gives
\[\E g_{\bmu}(s)=\eps\left[r-\sum_{\ell=2}^{r}\frac{\mu_{\ell-1}}{\mu_\ell}\right].\]
Decreasing \(\mu_1\) to \(\eps\) and increasing \(\mu_r\) to one can only improve this expression. The product of the remaining ratios is then \(\eps\). By the arithmetic--geometric mean inequality their sum is at least \((r-1)\eps^{1/(r-1)}\), with equality at the displayed geometric ladder. Substitution proves the formula.
\end{proof}

As \(r\to\infty\), this expression increases to \(\eps(1+\ln(1/\eps))=\OPT\). The exact formula characterizes this family, not the best communication protocol on every instance.
\section{Computation, finite populations, and scope}
\label{sec:implementation}

For fixed thresholds, the online protocol consists of one message to a clearing rule and one returned service level.  The admission kernel is computed offline from \(F\), \(\bq\), and the thresholds. 

\begin{proposition}[Finite LP and stochastic implementation]\label{prop:implementation} 
For a fixed one-threshold instance, the LP has one variable for each positive-mass referral profile and \(m\) capacity constraints.  For a fixed \(r\)-level ladder, it has at most \(r\) variables per positive-mass joint bucket profile and \(m\) capacity constraints.  Any feasible solution induces a stochastic kernel whose expected coordinate loads satisfy the original supplies. 
\end{proposition} 

For an empirical distribution on \(n\) observed types, there are at most \(n\) positive-mass profiles.  The same LP can therefore be solved directly from profile counts.  Independent conditional admission implements its fractional solution ex ante.  If hard realized capacities are required, rounding each profile--level admitted count down and assigning disjoint individuals to the resulting counts preserves every conservative capacity constraint.  With welfare normalized per capita, the loss in the \emph{certified LP objective} is less than \(\mu_\ell/n\) for a rounded level-\(\ell\) decision, and hence at most \(D/n\) over \(D\) positive profile--level decisions.  This is not a bound on the loss of full actual welfare under arbitrary deterministic choices within a profile.  Uniform random disjoint assignment within each profile preserves the profile-only information restriction and gives the same bound in expectation relative to the original kernel's actual welfare: an individual's welfare gain from level \(\mu_\ell\) is between zero and \(\mu_\ell\).  Sharper dependent-rounding guarantees are separate from the information frontier. This support-sensitive observation concerns a fixed threshold or ladder; it does not assert polynomial dependence on \(m\), nor an end-to-end algorithm for optimizing over a continuum of thresholds.

Three boundaries matter when applying the theorems.

\begin{itemize} 
\item \textbf{Ex-ante resources.}  Capacities are expectations over the population and protocol randomness.  Ex-post inventory, queues, and arrivals need an online model. 
\item \textbf{Public planning distribution.}  The protocol may use \(F\) and \(\bq\) offline.  Estimating them from data introduces statistical error not analyzed here. 
\item \textbf{Non-strategic messages.}  Agencies report deterministic threshold comparisons of their observations.  Incentives, verification, multi-round communication, and formal privacy are outside the model. 
\end{itemize} 

Conservative accounting is essential only for the universal upper certificates: it replaces actual cost \((\mu-w_j)_+\) by \(\mu\). Every constructed protocol is therefore feasible under actual top-up costs. The equal-revenue lower bounds use \(q_1=0\).  Both accounting conventions then prohibit assigning a type \((s,0)\) to a level above \(s\), so the stated lower bounds remain valid under either convention even though their numerical charges need not coincide.
\section{Conclusion}
\label{sec:conclusion}

Complementarity turns local indistinguishability into a global welfare loss. With no communication, even an optimal distribution-aware product policy can lose \(\Theta(m)\) or \(\Theta(1/\OPT)\).  One common threshold repairs the worst obstruction up to the best possible logarithm for that class. Ordered thresholds then convert survival-curve rectangles into a staircase: \(\Theta(\log(1/\OPT))\) levels, encoded in \(\Theta(\log\log(1/\OPT))\) bits per agency, are sufficient and necessary for constant-factor recovery within ordered referrals.

The operative distinction is not local versus centralized control.  It is whether agencies share enough structured information to identify which local deficits belong to the same complementary bottleneck.
\newpage

{\fontsize{11pt}{13.2pt}\selectfont
\bibliographystyle{plainnat}
\bibliography{references}

\paragraph{Generative AI Use.}
Codex assisted with proof auditing and literature-search suggestions. Responsibility for the final manuscript rests with the author.

\clearpage}

\normalsize
\appendix
\section{Proofs}
\label{app:proofs}

\paragraph{Attainment of the full-information benchmark.}
The supremum defining \(\OPT\) is attained. To see this, introduce a utility variable \(v\) and work in \(L^2(F)^{m+1}\) with constraints \(0\le X_j\le1-w_j\), \(\E X_j\le q_j\), \(0\le v\le1\), and \(v\le w_j+X_j\) almost surely. This feasible set is non-empty, bounded, convex, and norm-closed, thus weakly compact. The continuous linear functional \(\E v\) attains its maximum there. Replacing \(v\) by the induced utility can't decrease its value and preserves feasibility. Representatives can be chosen measurable and modified on null-sets to satisfy the pointwise bounds. This justifies our use of \(X^\star\) throughout proofs.

\subsection{Interim reduction}

\begin{proof}[Proof of \Cref{lem:interim}]
Let \(R\) collect all private and shared randomization, independent of the type. A randomized local-product rule has \(X_j(\bw;R)=x_j(w_j;R)\). Define
\[\bar x_j(w_j)=\E_R[x_j(w_j;R)]\quad\text{and}\quad\bar X(\bw)=(\bar x_1(w_1),\ldots,\bar x_m(w_m)).\]
Linearity preserves expected loads:
\[\E_F[\bar x_j(w_j)]=\E_{F,R}[x_j(w_j;R)]\le q_j.\]
The no-waste bounds are preserved under expectation. For each fixed \(\bw\), the map
\[\bx\longmapsto\min\{1,w_1+x_1,\ldots,w_m+x_m\}\]
is concave. Jensen's inequality thus gives
\[u^{\bar X}(\bw)\ge\E_R[u^{X(\cdot;R)}(\bw)].\]
Taking expectation over \(F\) proves the claim.
\end{proof}

\subsection{Local-information losses}

\begin{proof}[Proof of \Cref{thm:local-loss}]
We treat the two families separately.

\paragraph{Balanced target-decoy family.}
Full information gives one unit of resource \(i\) if the type is \(\bw^{(i)}\); and $0$ otherwise. Expected load on each coordinate is \(q\), and welfare is \(mq\). Conversely, for any feasible allocation,
\[\begin{aligned}\USW(X)&=\sum_{i=1}^m q\,u^X(\bw^{(i)})+(1-mq)u^X(\bzero)\\
&\le\sum_{i=1}^m\left(qX_i(\bw^{(i)})+(1-mq)X_i(\bzero)\right)\\
&\le\sum_{i=1}^m q_i=mq.\end{aligned}\]
The first inequality uses \(u^X(\bw^{(i)})\le X_i(\bw^{(i)})\) and \(u^X(\bzero)\le X_i(\bzero)\) for every \(i\); counting the non-negative decoy term \(m\) times only enlarges the bound. The second inequality is coordinate-wise feasibility, as all other target types have baseline $1$ in coordinate \(i\) and as such get no useful top-up there. Thus \(\OPT=mq\).

Now consider a deterministic local-product rule. Agency \(j\) sees \(w_j=0\) on its target and on the decoy, a cell of mass
\[q+(1-mq)=1-(m-1)q.\]
If \(a_j=x_j(0)\), feasibility implies
\[a_j\le a:=\frac{q}{1-(m-1)q}.\]
Every target's utility is at most its missing-coordinate top-up, and the decoy's utility is at most \(\min_j a_j\); so every type has utility at most \(a\) and \(\DEC\le a\). Set \(x_j(0)=a\) and \(x_j(1)=0\), for all \(j\), to achieve equality. Therefore
\[\frac{\OPT}{\DEC}=m\bigl(1-(m-1)q\bigr).\]
Assumption \(q\le1/(2m)\) places this quantity in \([(m+1)/2,m)\).

\paragraph{Scarce-complement family.}
Full information tops-up coordinates \textit{one} and \textit{two} of the target to one. This uses expected load \(q\) in each coordinate and achieves welfare \(q\). For every feasible allocation,
\[\USW(X)\le\E[w_1+X_1]=\E[X_1]\le q,\]
so \(\OPT=q\).

Under a local-product rule, agencies \textit{one} and \textit{two} see $0$ on the entire population, so \(x_1(0),x_2(0)\le q\). Agency \textit{three} sees $0$ only on the decoy, so \(x_3(0)\le q_{\min}/(1-q)\). Let \(a=\min\{x_1(0),x_2(0)\}\le q\). The target and decoy contributions are bounded by
\[q a\le q^2\quad\text{and}\quad(1-q)\min\left\{q,\frac{q_{\min}}{1-q}\right\}=\min\{q(1-q),q_{\min}\},\]
respectively. This proves the upper bound on \(\DEC\). It is attained by
\[x_1(0)=x_2(0)=q,\qquad x_3(0)=\min\left\{q,\frac{q_{\min}}{1-q}\right\},\]
with all other top-ups $0$. The exact formula follows. If \(q_{\min}\le q^2\), then \(q^2\le\DEC\le2q^2\), which gives \(\OPT/\DEC=\Theta(1/q)=\Theta(1/\OPT)\).

Finally, \Cref{lem:interim} transfers both upper-bounds to randomized local-product rules (with type-independent randomness).
\end{proof}

\subsection{One-threshold recovery and obstruction}

\begin{proof}[Proof of \Cref{lem:charging}]
For a fixed \(\mu>0\), define
\[z_S=\Prob[S^\mu(\bw)=S,\ U(\bw)\ge\mu]\qquad(S\subseteq[m]).\]
Clearly \(0\le z_S\le p_S^\mu\). For any coordinate \(j\),
\[\begin{aligned}
    \mu\sum_{S\ni j}z_S
    &=\E\!\left[\mu\ind\{U\ge\mu,\ w_j<\mu\}\right]\\
    &\le\E\!\left[X_j^\star(\bw)\ind\{U\ge\mu,\ w_j<\mu\}\right]\\
    &\le \E[X_j^\star(\bw)]\le q_j.
\end{aligned}\]
The first inequality follows point-wise from
\[U\ge\mu,\quad w_j<\mu\quad\Longrightarrow\quad X_j^\star\ge u^\star-w_j\ge2\mu-w_j>\mu.\]
Thus \((z_S)_S\) is feasible for \eqref{lp:one-objective}-\eqref{lp:one-capacity}, and its objective is
\[\mu\sum_S z_S=\mu\Prob[U\ge\mu].\]
\end{proof}

\begin{proof}[Proof of \Cref{thm:one-bit}]
Let \(\delta=\OPT/4\). As \(\E U=\OPT/2\) and \(U\in[0,1/2]\),
\begin{equation}
\E[U\ind\{U\ge\delta\}]\ge\E U-\delta=\frac{\OPT}{4}.\label{eq:retained-U}
\end{equation}
Let
\[A=\sup_{\delta\le\mu\le1/2}\mu\Prob[U\ge\mu].\]
This supremum is attained: \(\mu\mapsto\Prob[U\ge\mu]\) is upper semi-continuous, and the interval is compact. The tail-integral identity on the retained range gives
\begin{align*}
\E[U\ind\{U\ge\delta\}]
&=\delta\Prob[U\ge\delta]+\int_\delta^{1/2}\Prob[U\ge t]\,dt\\
&\le A+A\int_\delta^{1/2}\frac{dt}{t}\\
&=A\left(1+\ln\frac{1/2}{\delta}\right)=A\bigl(1+\ln(2/\OPT)\bigr).
\end{align*}
Together with \eqref{eq:retained-U}, this gives:
\[A\ge\frac{\OPT}{4(1+\ln(2/\OPT))}.\]
By \Cref{lem:charging}, \(R(\mu)\ge\mu\Prob[U\ge\mu]\) for every \(\mu\), so \(\REFone\ge A\).
\end{proof}

\begin{proof}[Proof of \Cref{thm:one-threshold-lower}]
The tail formula gives
\[\E s=\int_0^1\Prob[s\ge t]\,dt=\eps+\int_\eps^1\frac{\eps}{t}\,dt=\eps\bigl(1+\ln(1/\eps)\bigr).\]
Full information allocates \(s\) units of resource-\textit{two} and $0$ of resource-\textit{one}, giving every type utility \(s\) at expected cost \(q_2=\E s\). Conversely, \(q_1=0\) and non-negative top-ups imply \(X_1=0\) almost surely, so utility is at most coordinate-\textit{one}'s \textbf{baseline} \(s\). Hence \(\OPT=\E s\).

Fix a threshold \(\mu\in(0,1]\). If \(\mu\le\eps\), each served type gets at most \(\mu\), so total welfare is at most \(\mu\le\eps\). If \(\mu>\eps\), a type with \(s<\mu\) would need the strictly positive coordinate-\textit{one} top-up \(\mu-s\). But since \(q_1=0\), such a type must have $0$ admission probability. Only the mass with \(s\ge\mu\) can be served, and so:
\[\USW\le\mu\Prob[s\ge\mu]=\mu\frac{\eps}{\mu}=\eps.\]
This argument uses actual top-up cost and applies to every admission kernel measurable with respect to the threshold profile.

At \(\mu=\eps\), admit everyone. Coordinate-\textit{one} already satisfies the threshold; coordinate-\textit{two} spends \(\eps\le\E s=q_2\). Welfare is exactly \(\eps\). Thus the upper-bound is tight.
\end{proof}

\subsection{Ordered ladders}

\begin{proof}[Proof of \Cref{lem:staircase}]
For each type, let
\[\ell(\bw)=\max\bigl(\{\ell\in[r]:\mu_\ell\le U(\bw)\}\cup\{0\}\bigr).\]
For every positive-mass joint bucket profile \(a\), set
\[z_{a\ell}=\Prob[a^{\bmu}(\bw)=a,\ \ell(\bw)=\ell]\qquad(\ell\in[r]).\]
The profile-mass constraints hold as each type is assigned to at most one positive level. If \(j\in S_\ell(a)\) on an assigned type, then \(w_j<\mu_\ell\le U\), and hence \(X_j^\star>\mu_\ell\) by the half-utility inequality. Thus,
\begin{align*}
&\sum_a\sum_{\ell:j\in S_\ell(a)}\mu_\ell z_{a\ell}\\
&\qquad
=
\E\!\left[
        \mu_{\ell(\bw)}\ind\{\ell(\bw)>0,\ w_j<\mu_{\ell(\bw)}\}
    \right]
\le
\E[X_j^\star]
\le 
q_j.
\end{align*}
The constructed solution is feasible for the ladder LP and has objective
\[\sum_{a,\ell}\mu_\ell z_{a\ell}=\E[\mu_{\ell(\bw)}]=\E[g_{\bmu}(U)].\]
\end{proof}

\begin{proof}[Proof of \Cref{thm:geometric-ladder}]
Since \(0<\OPT\le1\), we have \(0<\delta\le1/4\). By definition of \(L\), the first \(L\) levels \(2^{\ell-1}\delta\) lie strictly below \(1/2\), while the last level is \(\mu_{L+1}=1/2\). Thus the displayed ladder is strictly increasing.

For any \(u\in[\delta,1/2]\), the largest ladder point below \(u\) is at least \(u/2\). The geometric ladder rounds every \(u\in[\delta,1/2]\) down by a factor of at most $2$: the selected level is no smaller than \(u/2\). Therefore,
\[g_{\bmu}(U)\ge\frac{U}{2}\ind\{U\ge\delta\}.\]
As in \eqref{eq:retained-U},
\[\E[U\ind\{U\ge\delta\}]\ge\E U-\delta=\frac{\OPT}{4}.\]
Applying \Cref{lem:staircase} gives
\[R(\bmu)\ge\E[g_{\bmu}(U)]\ge\frac12\cdot\frac{\OPT}{4}=\frac{\OPT}{8}.\]
There are \(r=1+\lceil\log_2(2/\OPT)\rceil\) levels and \(r+1\) local buckets, so the exact uplink message length is \(\lceil\log_2(r+1)\rceil\), which is \(O(\log\log(1/\OPT))\) as \(\OPT\downarrow0\).
\end{proof}

\begin{proof}[Proof of \Cref{thm:ladder-lower}]
Consider any ordered ladder \(0<\mu_1<\cdots<\mu_r\le1\) on the equal-revenue family and set \(\mu_0=0\). As \(q_1=0\), non-negativity and ex-ante feasibility imply that coordinate-\textit{one}'s top-up is $0$ almost surely. A type \((s,0)\) can consequently be assigned only to levels \(\mu_\ell\le s\). Its welfare is at most
\[g_{\bmu}(s)=\max\bigl(\{\mu_\ell:\mu_\ell\le s\}\cup\{0\}\bigr).\]
The standard staircase expansion gives:
\begin{equation}
\E[g_{\bmu}(s)]=\sum_{\ell=1}^r(\mu_\ell-\mu_{\ell-1})\Prob[s\ge\mu_\ell].\label{eq:staircase-expansion}
\end{equation}
For levels at or above \(\eps\),
\[(\mu_\ell-\mu_{\ell-1})\Prob[s\ge\mu_\ell]\le\mu_\ell\frac{\eps}{\mu_\ell}=\eps.\]
The increments at levels strictly below \(\eps\) telescope to less than \(\eps\). If there are \(k\ge1\) such levels, the total in \eqref{eq:staircase-expansion} is at most \(\eps+(r-k)\eps\le r\eps\); if there are none, each of the \(r\) terms is at most \(\eps\). Thus every feasible kernel has welfare at most \(r\eps\).

If the protocol obtains at least \(c\OPT\) for a fixed \(c>0\), then
\[r\eps\ge c\eps\bigl(1+\ln(1/\eps)\bigr),\qquad\text{so}\qquad r\ge c\bigl(1+\ln(1/\eps)\bigr).\]
Under the full-alphabet, fixed-length convention of \Cref{sec:model}, an \(r\)-level ordered protocol uses \(r+1\) local buckets and so \(\lceil\log_2(r+1)\rceil\) bits per agency. Finally,
\[\ln(1/\OPT)=\ln(1/\eps)-\ln\bigl(1+\ln(1/\eps)\bigr)=\Theta(\ln(1/\eps))\]
as \(\eps\downarrow0\). The claimed level and bit lower-bounds then follow.
\end{proof}

\subsection{Implementation}

\begin{proof}[Proof of \Cref{prop:implementation}]
For one threshold, profiles with \(p_S^\mu=0\) can be omitted from \eqref{lp:one-objective}-\eqref{lp:one-capacity}; this leaves one variable per positive-mass profile and \(m\) capacity constraints, in addition to box constraints. Given a feasible solution, the clearing rule admits conditional on profile \(S\) with probability \(z_S/p_S^\mu\). Its \emph{actual} expected load in coordinate \(j\) is
\[\sum_{S\ni j}z_S\,\E[(\mu-w_j)_+\mid S,\text{ admitted}]\le\mu\sum_{S\ni j}z_S\le q_j.\]
Conditional on the reported profile, admission is independent of raw coordinates, so this is a valid profile kernel.

For a ladder, omit profiles with \(p_a=0\). There are at most \(r\) variables \(z_{a\ell}\) per remaining profile. Conditional on profile \(a\), choose level \(\ell\) with probability \(z_{a\ell}/p_a\) and choose $0$ service with the remaining probability. For each coordinate, \emph{actual} expected top-up is at most
\[\sum_a\sum_{\ell:j\in S_\ell(a)}\mu_\ell z_{a\ell}\le q_j\]
by \eqref{lp:ladder-capacity}. This proves the stated size and stochastic implementation claims.
\end{proof}

\end{document}